\documentclass[letterpaper, 10pt, conference]{ieeeconf}      

\usepackage{xcolor}
\usepackage{cite}
\usepackage{array}
\usepackage{amsmath,amssymb,amsfonts}
\usepackage{mathtools}
\usepackage{cuted}
\usepackage{graphicx}
\usepackage{empheq}

\IEEEoverridecommandlockouts       
\usepackage[ruled,vlined]{algorithm2e}

\title{\LARGE \bf
Distributed two-time-scale methods over clustered networks
}

\author{ Thiem V. Pham$^{1}$ \;\; Thinh T. Doan$^{2}$\;\;   Dinh Hoa Nguyen$^{3}$
\thanks{$^{1}$CReSTIC, University of Reims Champagne-Ardenne, France. Email: {\tt\small van-thiem.pham@etudiant.univ-reims.fr}}%
\thanks{$^{2}$Department of Electrical and Computer Engineering, Virginia Tech, USA. Email: {\tt\small thinhdoan@vt.edu}}%
\thanks{$^{3}$International Institute for Carbon-Neutral Energy Research (WPI-I$^2$CNER), and Institute of Mathematics for Industry (IMI), Kyushu University, 744 Motooka, Nishi-ku, Fukuoka 819-0395, Japan. Email:         {\tt\small hoa.nd@i2cner.kyushu-u.ac.jp}}%
}

\newcommand{\Rset}{\mathbb{R}}

\newcommand{\Ccal}{{\cal C}}

\newcommand{\Ecal}{{\cal E}}

\newcommand{\Gcal}{{\cal G}}

\newcommand{\Ncal}{{\cal N}}

\newcommand{\Vcal}{{\cal V}}

\newcommand{\Ibf}{{\bf I}}

\newcommand{\Vbf}{{\bf V}}
\newcommand{\Wbf}{{\bf W}}
\newcommand{\Xbf}{{\bf X}}
\newcommand{\Ybf}{{\bf Y}}

\newcommand{\1}{{\mathbf{1}}}

\newcommand{\xbar}{{\bar{x}}}

\newtheorem{lem}{Lemma}
\newtheorem{thm}{Theorem}

\newtheorem{assump}{Assumption}
\newtheorem{remark}{Remark}

\begin{document}

\maketitle
\thispagestyle{empty}
\pagestyle{empty}

\begin{abstract}
In this paper, we consider consensus problems over a network of nodes, where the network is divided into a number of clusters. We are interested in the case where the communication topology within each cluster is dense as compared to the sparse communication across the clusters. Moreover, each cluster has one leader which can communicate with other leaders in different clusters. The goal of the nodes is to agree at some common value under the presence of communication delays across the clusters.      

Our main contribution is to propose a novel distributed two-time-scale consensus algorithm, which pertains to the separation in network topology of clustered networks. In particular, one scale is to model the dynamic of the agents in each cluster, which is much faster (due to the dense communication) than the scale describing the slowly aggregated evolution between the clusters (due to the sparse communication). We prove the convergence of the proposed method in the presence of uniform, but possibly arbitrarily large, communication delays between the leaders. In addition, we provided an explicit formula for the convergence rate of such algorithm, which characterizes the impact of delays and the network topology. Our results shows that after a transient time characterized by the topology of each cluster, the convergence of the two-time-scale consensus method only depends on the connectivity of the leaders. Finally, we validate our theoretical results by a number of numerical simulations on different clustered networks.

\end{abstract}

\section{Introduction}
Clustered network of agents is a specific type of multi-agent systems, where the whole network is divided into distinct clusters, and usually the connection structure in each cluster is denser, while the inter-cluster connection is sparser. Each cluster might contains smaller clusters inside, resulting in a hierarchical and multi-layer clustered network. This type of system can be found in a variety of application domains including energy systems \cite{Romeres13}, robotics \cite{Magnus2020}, biological and chemical engineering \cite{Bleibel18,Leiser17}, social networks \cite{Proskurnikov17}, brain science \cite{ADas19}, epidemic \cite{Leitch19}, etc., and has been a timely research topic in network science \cite{Barabasi-book06,barabasi2016network}. 

As an example, power and energy systems are large-scale systems composing of many subsystems inside, each of them can be regarded as a cluster \cite{Romeres13}.   
In another example of social networks, the opinion of each individual continuously evolves with respect to the views of the members belonging to its community in order to achieve a common agreement. In some specific conditions, at specific instants, one individual in each community (called a leader) can change its opinion by exchanging with other leaders outside its community. They will reset their opinion taking into account the ones of other leaders. These inter-cluster interactions can be considered as resets of the opinions \cite{Mor2016}. 

Hitherto, the existing literature on clustered networks of agents aims at either exploring how network structures affect to the controllability (e.g.,\cite{Rahimian13}), observability \cite{Mousavi19}, and control performances of the network \cite{Siami16}, or exploiting special properties of such networks for enhancing the overall network robustness, resiliency, etc. \cite{Pasqualetti13}.  It is worth emphasizing that in all researches on agent networks, network convergence is one of the most fundamental problems, whether infinitely or in finite-time. Network convergence could be significantly altered by latency, both intra-cluster and inter-cluster. To study clustered networks, singular perturbation theory is one popular approach, see e.g., \cite{Romeres13,Boker16,Awad19,JChow85,Biyik08,XCheng18,SMartin16}. Under this approach, aggregated and reduced order models were derived, where the fast dynamics inside clusters is ignored, or lumped into that of the slow dynamics occurring across clusters. As such, {\it the obtained results mostly depend on what happen between clusters.} 

Note also that the existing literature on distributed algorithms (stabilization, consensus, formation, etc.) for clustered networks has focused only on the network asymptotic convergence of these algorithms, while their {\it convergence rates are missing}, see e.g., \cite{Mor2016,Rejeb2015,Pham2020c,Pham2019f}, in addition to the aforementioned researches. The study in \cite{Mor2016} showed the existence of a positive decay rate to guarantee the overall network asymptotic consensus. The event-triggered resets defined for each cluster leading to asynchronous reset sequences were proposed in \cite{Rejeb2015}. In our recent work \cite{Pham2020c}, a robust formation controller design was proposed for clustered networks of unmanned aerial vehicles, but again the convergence was only asymptotic.  

In this research, motivated by the fact that the convergence speed inside clusters is usually much faster than that between clusters, due to denser connection structures and shorter communication distances, our goal is to derive explicitly the network convergence rate in accordance to inter-cluster delays and network structures. Several recent works have investigated the inter-cluster time delays, e.g. \cite{Magnus2020,Boker16,Liu2017c}, to achieve asymptotic network convergence. The work \cite{Magnus2020} utilized passivity theory for the cooperative control of two clusters of robots over a very long distance which naturally incurs a very large time delay between such two clusters. The study \cite{Boker16} designed state-feedback controllers for clustered networks using singular perturbation theory. Nevertheless, in all of above researches, {\it no convergence rate was considered.} 

On the other hand, it can be observed that multi-time-scale algorithms are suitable approaches for studying clustered networks, due to their nature of different time scales inside clusters and across clusters. Therefore, the current research proposes a distributed two-time-scale consensus algorithm for clustered networks with inter-cluster time delays, where a faster consensus protocol is employed intra-cluster and a slower consensus law is accounted for inter-cluster time delays. Note that several two-time-scale methods have been proposed in the literature, e.g. \cite{Doan2017,DoanBS2018_ACC, Doan2018a, Doan,Doan2020, DoanM2019_Allerton, SMukherjee19,Dalal2017a,pmlr-v99-karimi19a,GuptaSY2019_twoscale,MokkademP2006,Konda2004}, for different problems in machine learning and reinforcement learning, however they are different from the one proposed in this paper.

\textbf{Contribution.}
Our main contribution is to propose a novel distributed two-time-scale consensus algorithm, which pertains to the separation in network topology of clustered networks. In particular, one scale is to model the dynamic of the agents in each cluster while one scale is to present the slowly aggregated evolution between the leaders of the clusters. We prove the convergence of the proposed method in the presence of uniform, but possibly arbitrarily large, communication delays between the leaders. In addition, we provided an explicit formula for the convergence rate of such algorithm, which characterizes the impact of delays and the network topology. Our results shows that after a transient time characterized by the topology of each cluster, the convergence of the two-time-scale consensus method only depends on the connectivity of the leaders. Moreover, the results in this paper complements for the existing consensus literature over clustered networks where such a formula of convergence rates is missing. Finally, we validate our theoretical results by a number of numerical simulations on different clustered networks.


The rest of this paper is organized as follows. The problem setting and the proposed algorithm are formalized in Section \ref{sec:prob}. Next, the main result of this paper is presented in Section \ref{sec:results}. Finally, a number of numerical simulations are provided in Section \ref{sec:simulation} to illustrate the theoretical results of this paper. 
\section{Distributed two-time-scale consensus methods}
\label{sec:prob}
\begin{figure}[htpb!]
    \centering
    \includegraphics[width=\linewidth]{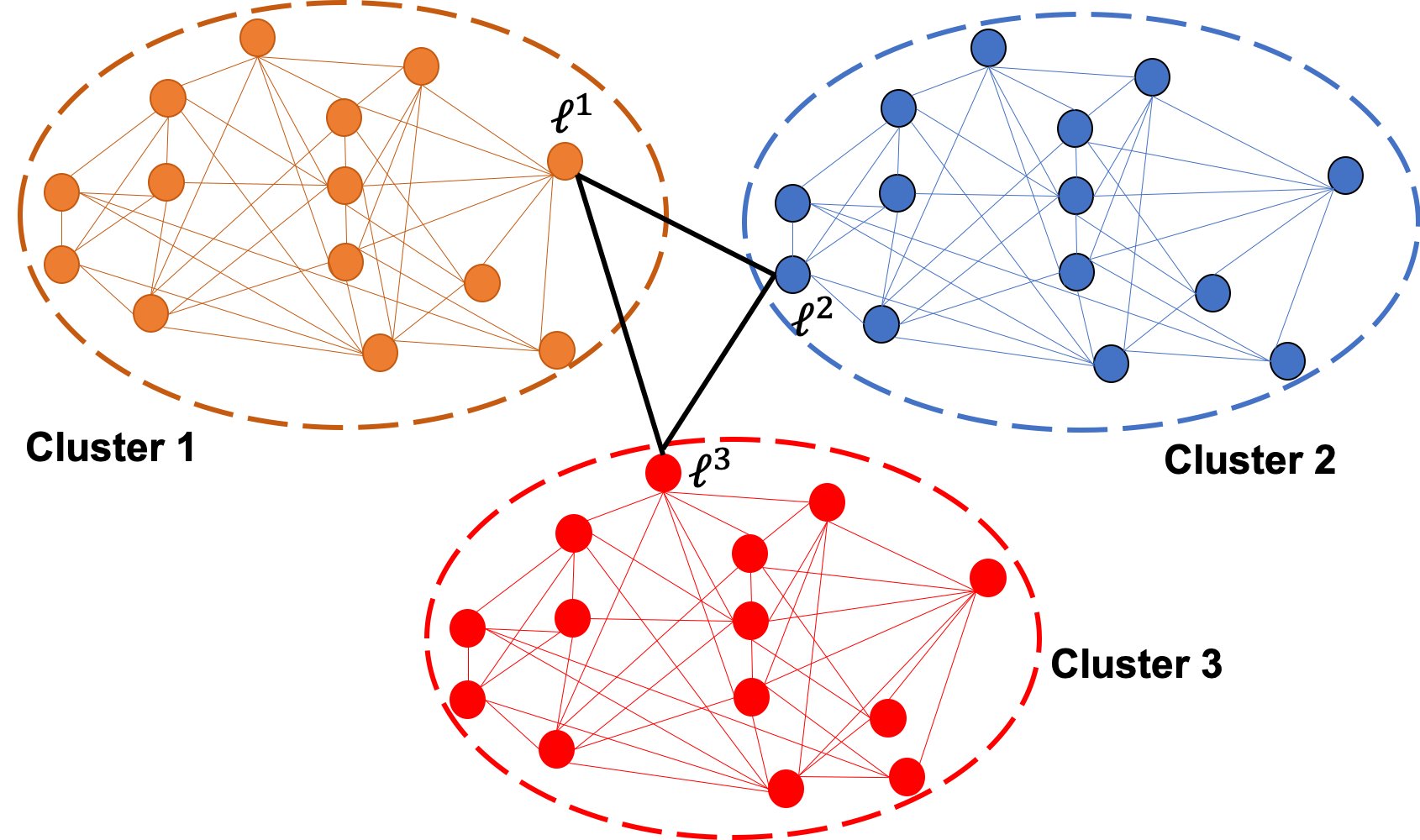}
    \caption{{\small A network partitioned into $3$ densely connected clusters}.}
    \label{fig:cluster}
\end{figure}
We consider a network of $N$ nodes divided into $r$ disjoint connected clusters $\Ccal^{a}$, $a=1,\ldots,r$. The communication pattern between nodes within each cluster $\Ccal^{a}$ is modeled by a densely connected and undirected graph $\Gcal^{a} = (\Vcal^{a},\Ecal^{a})$, where $\Vcal^{a}$ is the set of node indexes and $\Ecal^{a}$ is the set of edges. Thus, we have $\Vcal^{a}\cap\Vcal^{b} = \emptyset$ for all $a\neq b$ and $\sum_{a=1}^{r}|\Vcal^{a}| = N$, where $|\Vcal^{a}|$ denotes the cardinality of the set $\Vcal^{a}$. In addition, the communication between the nodes in different clusters is described by a sparse time-varying graph $\Gcal^{C}(k) = (\Vcal^{C},\{\Ecal^{C}(k)\})$, that is, $(i,j)\in\Ecal^{C}(k)$ if and only if the nodes $i$ and $j$ are in different clusters and there exists an edge between them at time $k$. We  assume that $|\Vcal^{a}\cap\Vcal^{C}| = 1$, $a=1,\ldots,r$, i.e., each cluster has only one node, called \textit{the leader} of $\Ccal^{a}$, that is connected to (some) leaders of other clusters. Similarly, we name other nodes as followers. For simplicity and notational convenience, we assume that all the followers in each cluster is connected to its leader. However, we note that our method derived later can work for more general setting. Finally, the $N$ nodes are connected by a graph $\Gcal = (\Vcal,\Ecal)$, where $\Vcal = \cup_{a=1}^{r}\Vcal^{a}$ and $\Ecal = \cup_{a=1}^{r}\Ecal^{a}\cup \Ecal^{C}$.

An illustrative example of clustered networks is given in Fig. \ref{fig:cluster} where the nodes are divided into $r$ clusters. The connection between nodes in each cluster is modeled by a densely connected graph, while each cluster $\Ccal^{a}$ has one leader $\ell^{a}$, $a = 1,\ldots,r$, connected to other leaders in other clusters (described by black edges).  A concrete motivating example for this problem is coordinated control of multi-robot systems across long distances \cite{Magnus2020}. As the spaces between robots grow they experience larger delays in their communication. The presence of delays, if not properly addressed, may lead to oscillatory behaviors and even instabilities among the robots. To handle the impact of delays, distributed consensus algorithm based on the so-called scattering transformation is proposed in \cite{Magnus2020}. By utilizing passivity theory, the authors can show an asymptotic convergence of their proposed method under the presence of communication delays. Our focus in this paper is to understand the rate of convergence of consensus algorithm over clustered networks, which is missing in \cite{Magnus2020}. Indeed, while passivity theory can help to study the stability (or asymptotic convergence) of consensus algorithms under delays, it may not be applicable to study the finite-time convergence. Therefore, we propose a novel distributed two-time-scale  algorithm for solving consensus problems under delays over clustered networks. Moreover, we carefully characterize the finite-time performance of our proposed methods discussed in details in Section \ref{sec:results}. 

Our proposed method, formally stated in Algorithm \ref{alg:dist_two_time_scale}, is explained as follows. Suppose that each node $i\in\Vcal$ is initialized with some arbitrary value $p_{i}$. This variable may encode the initial local information at each node in $\Gcal$, e.g., the position and velocity of a robot in networked formation control problems. The goal of the nodes is to cooperatively agree at some common value, i.e., they achieve a consensus. Since the graphs $\Gcal^{a}$ are dense while $\Gcal^{C}(k)$ is sparse, information shared between the nodes within each cluster is mixed much faster than the one between the clusters. For example, in Fig.\ \ref{fig:cluster} the rate of information sent from $\ell^{1}$ to its followers in $\Ccal^{1}$ is much faster than the time it gets to the nodes in $\Ccal^{r}$. To model this difference, we propose a novel distributed two-time-scale consensus method, where the updates of the followers in any cluster is implemented at a faster time-scale as compared to the one between the leaders of the clusters. This is described by the use of two different step sizes in Eqs.\ \eqref{alg:follower} and \eqref{alg:leader}.  

\begin{algorithm}
\SetAlgoLined
\textbf{Initialization:} Each follower $i\in\Vcal^{a}$ maintains $x_{i}^{a}$ and each leader $a$ maintains $x_{\ell}^{a}$,  $\forall a\in[1,r]$ 

Each node $i\in\Vcal$ initializes its variable at a constant $p_{i}$

The nodes initialize proper step sizes $\beta\ll\gamma\in(0,1)$.

 \For{k=0,1,2,...}{
  \For{each cluster $a = 1,\ldots,r$}{
    \For{each follower $i\in\Vcal^{a}$}{ 1. Receive $x_{\ell}^{a}(k)$ from its leaders
    
    2. Exchange $x_i^a(k)$ with neighbors $j\in\mathcal{N}_i^{a}$
    
    3. Implement
    \begin{align}
        \hspace{-0.5cm} x_{i}^{a}(k+1)\! &=\! (1-\gamma)\!\!\sum_{j\in\Ncal_{i}^{a}}w_{ij}^{a}x_{j}^{a}(k) + \gamma x_{\ell}^{a}(k).\label{alg:follower}
    \end{align}
    }
    \underline{Leader} $\ell^{a}$: Exchange $x_{\ell}^{a}$ to other leaders $b\in\Ncal_{a}^{C}(k)$ and update
        \begin{align}
        x_{\ell}^{a}(k+1) &= (1-\beta)x_{\ell}^{a}(k)\notag\\ 
        &\quad + \beta \sum_{b\in\Ncal_{a}^{C}(k)}v_{ab}(k)x_{\ell}^{b}(k-\tau).\label{alg:leader}
        \end{align}
  }
 }
\caption{Distributed two-time-scale consensus methods under delays}
\label{alg:dist_two_time_scale}
\end{algorithm}

In particular, for each cluster $\Ccal^{a}$ each follower $i\in\Vcal^{a}$ maintains  a variable $x_{i}^{a}$ and the leader maintains $x_{\ell}^{a}$, both are set at their initial values. The nodes in each $\Ccal^{a}$, $a = 1,\ldots,r$, then consider the updates in \eqref{alg:follower} and \eqref{alg:leader}, where $\Ncal_{i}^{a}$ is the neighboring set of node $i$ in the cluster $\Ccal^{i}$ and $\Ncal_{a}^{C}(k)$ is the neighboring set of leader $a$ in the graph $\Gcal^{C}(k)$ of the clusters at time $k$. The step sizes $\gamma$ and $\beta$ are in $(0,1)$. Moreover, $w_{ij}^{a}$ and $v_{ab}(k)$ are some positive (time-varying) weights, which will be specified shortly. Here, Eq.\ \eqref{alg:follower} is a consensus step between the followers in each cluster $\Ccal^{a}$ and Eq.\ \eqref{alg:leader} is the consensus step between the leaders of the clusters. We allow that the updates of the followers may take into account the value of its leaders but not vice versa.  Finally, the constant $\tau$ in \eqref{alg:leader} represents the communication delays in the information exchange between the leaders. 

In Eqs.\ \eqref{alg:follower} and \eqref{alg:leader}, we use two different step sizes $\gamma,\beta$ to represent the difference in information propagation between the nodes within each cluster and across different clusters. Indeed, since $\gamma$ is associated with the fast-time scale of the followers' updates in each cluster we choose $\gamma \gg\beta$, which corresponds to the slow-time scale of the updates between the clusters. As shown in our numerical experiments in Section \ref{sec:simulation}, the followers' iterates move toward to the leaders' values, which is slowly pushed to a consensus value through step \eqref{alg:leader}. Finally, the step size $\beta$ is also chosen properly to handle the delays as considered in the previous work \cite{Doan2017}.  

\begin{remark}
Recall that in \eqref{alg:follower} we assume that the followers in each cluster is connected to its leader. Such an assumption only helps to reduce the burden notation in our algorithm. Our two-time-scale approach, however, can be applied for solving the consensus problem over general cluster networks without requiring this assumption.  
\end{remark}

\section{Main results}\label{sec:results}
In this section, we analyze the convergence properties of the proposed distributed two-time-scale methods under delays presented in the previous section. Specifically, our results show that under some proper choice of step size $\gamma\gg\beta$ the nodes in the network $\Gcal$ reach a consensus at a rate $\beta/\gamma$. In addition, we provide an explicit formula to show the dependence of this convergence on the network topology and the constant delays $\tau$. By using the two-time-scale approach, we can show that the convergence of the followers within each cluster $\Ccal^a$ only depends on the topology of $\Gcal^
{a}$ while the convergence of the leaders only depends on the sparse graph $\Gcal^{C}$. Since the former convergence happens much faster than the latter, we observe that after a transient time the convergence of Algorithm \ref{alg:dist_two_time_scale} only depends on the connectivity of $\Gcal^{C}$.

We begin our analysis by introducing more notation. We denote by $\Wbf^{a} = [w_{ij}^{a}] \in \Rset^{|\Vcal^{a}|\times |\Vcal^{a}|}$ the weighted adjacency matrix corresponding to $\Gcal^{a}$ at cluster $\Ccal^{a}$. Similarly, let $\Vbf(k) = [v_{ab}(k)]\in \Rset^{r\times r}$ be the time-varying weighted adjacency matrix corresponding to the sequence of graphs $\{\Gcal_{C}(k)\}$ between the clusters. For convenience, we use the following notation
\begin{align*}
\Xbf^{a} \triangleq \left(\begin{array}{c}
(x_{1}^{a})^T    \\
\vdots\\
(x_{n}^{a})^T   
\end{array}\right) \in\Rset^{|\Vcal^{a}|\times d},\;\Xbf_{\ell} \triangleq \left(\begin{array}{c}
(x_{\ell}^{1})^T    \\
\vdots\\
(x_{\ell}^{r})^T   
\end{array}\right) \in\Rset^{r\times d}.   
\end{align*}
Using this notation, the matrix forms of \eqref{alg:follower} and \eqref{alg:leader} are given
\begin{align}
\begin{aligned}
\Xbf^{a}(k+1) &= (1-\gamma)\Wbf^{a}\Xbf^{a}(k) + \gamma \1 x_{\ell}^{a}(k)^T\\
\Xbf_{\ell}(k+1) &= (1-\beta)\Xbf_{\ell}(k) + \beta \Vbf(k)\Xbf_{\ell}(k-\tau),
\end{aligned}\label{alg:matrix_form} 
\end{align}
where we denote by $\1$ a vector with proper dimension whose entries are all equal to the constant $1$. Finally, let $\xbar^{a}\in\Rset^{d}$, $a = 1,\ldots,r$ and $\xbar_{\ell}\in\Rset^{d}$ be the average of the row vectors of $\Xbf^{a}$ and $\Xbf^{\ell}$, respectively, i.e., 
\begin{align*}
\xbar^{a} = \frac{1}{|\Vcal^{a}|}\sum_{i\in\Vcal^{a}}x_{i}^{a}\;\text{ and }\; \xbar_{\ell} = \frac{1}{r}\sum_{a=1}^{r}x_{\ell}^{a}.
\end{align*}
Next, we make an assumption on $\Wbf^{a}$, $a= 1,\ldots,r$, and $\Vbf(k)$ which is fairly standard in the consensus literature to guarantee the convergence of the nodes’ estimates to a consensus point \cite{Olfati-Saber2004}. The assumption given below also imposes a constraint on the communication between the followers at each cluster and between the leaders of the clusters, in which the nodes are only allowed to exchange messages with neighboring nodes, i.e., those directly connected to them. 
\begin{assump}\label{assump:W_follower}
$\Wbf^{a}$, $a = 1,\ldots,r$, is a doubly stochastic matrix, i.e., $\sum_{i}w_{ij}^{a} = \sum_{j}w_{ij}^{a} = 1.$ Moreover, $w_{ii}^{a} > 0$, $\forall i$, and $w_{ij}^{a} > 0$ if and ony if $(i,j)\in\Ecal^{a}$ otherwise $w_{ij}^{a} = 0$.
\end{assump}
\begin{assump}\label{assump:V_leader}
There exists a positive constant $\gamma$ such that $\Vbf(k)$ satisfies the following conditions for all $k\geq0$:
\begin{itemize}
\item[(a)] $v_{aa}(k)\geq\alpha,$ for all $a = 1,\ldots,r$.
\item[(b)] $v_{ab}(k)\in[\alpha,1]$ if $(a,b)\in\mathcal{N}_{a}^{C}(k)$ otherwise $v_{ab}(k)=0$.
\item[(c)] $\sum_{a=1}^{r} v_{ab}(k)=\sum_{b=1}^{r} v_{ab}(k)=1,$ for all $a,b$.
\end{itemize}
\end{assump}
In addition, we assume that the graph $\mathcal{G}^{C}(k)$ is connected at any time $k\geq0$, formally stated as follows.
\begin{assump}\label{assump:BConnectivity}
For all $k\geq 0$, $\mathcal{G}^{C}(k) = (\Vcal,\Ecal(k))$ is connected and undirected.
\end{assump}
By Assumption \ref{assump:W_follower} and since $\Gcal^{a}$ is connected, each $\Wbf^{a}$ has $1$ as the largest singular value. Let $\sigma_{a}$ be the second largest eigenvalue of $\Wbf^{a}$, which by the Perron-Frobenis theorem \cite{Ren2005} we have $\sigma_{a}\in (0,1)$. Similarly, we denote by $\sigma(\Vbf(k))$ the second largest singular value of $\Vbf(k)$. Furthermore, let $\delta_{C}$ be a parameter representing the spectral properties of the time-varying graph $\Gcal^{C}(k)$ defined as
\begin{align}
\delta_{C} = \max_{k\geq 0}\sigma(\Vbf(k)).\label{notation:delta}
\end{align}
Assumptions \ref{assump:W_follower} and \ref{assump:BConnectivity} imply that $\delta_{C}\in(0,1)$. We now present the main steps in our analysis. We note that some of the analysis below is quite standard in the existing literature. We include them in this paper for completeness. Our first result is to show that the followers in each cluster reach a consensus exponentially, stated in the following lemma.  

\begin{lem}\label{lem:consensus_follower}
The sequence $\{\Xbf^{a}(k)\}$ generated by \eqref{alg:matrix_form} satisfies
\begin{align}
\|\Xbf^{a}(k) - \xbar^{a}(k)^T\1\| \leq \left((1-\gamma)\sigma_{a}\right)^{k+1}\|\Xbf^{a}(0)\|,\; \forall a.   \label{lem:consensus_follower:ineq}
\end{align}
\end{lem}

\begin{proof}
Denote by $\Ybf^{a}\! =\! \Xbf^{a} - \1(\xbar^{a})^T=(\Ibf-\frac{1}{n}\1\1^T)\Xbf^{a}$. Since $\Wbf^{a}$ is a doubly stochastic matrix, using \eqref{alg:matrix_form} we have
\begin{align}
\xbar^{a}(k+1) &= \frac{1}{|\Vcal^{a}|}\sum_{i\in\Vcal^{a}}x^a_{i}(k+1)\notag\\ 
&= (1-\gamma)\xbar^{a}(k) + \gamma x_{\ell}^{a}(k).
\label{lem:consensus_follower:eq1a}     
\end{align}
By \eqref{alg:matrix_form} and \eqref{lem:consensus_follower:eq1a} we consider for each $a = 1,\ldots,r$
\begin{align*}
\Ybf^{a}(k+1) &= \Xbf^{a}(k+1) - \1\xbar^{a}(k+1)^T\\
&= (1-\gamma)\Wbf^{a}\Xbf^{a}(k) - (1-\gamma)\1\xbar^{a}(k)^T\notag\\ 
&= (1-\gamma)\Wbf^{a}\Ybf^{a}(k),
\end{align*}
which by using $\|\Ybf^{a}\|\leq \|\Xbf^{a}\|$ yields \eqref{lem:consensus_follower:ineq}, i.e.,
\begin{align*}
\|\Ybf^{a}(k+1)\| &= \|(1-\gamma)\Wbf^{a}\Ybf^{a}(k)\| \leq (1-\gamma)\sigma_{a}\|\Ybf^{a}(k)\|\notag\\
&\leq \left((1-\gamma)\sigma_{a}\right)^{k+1}\|\Ybf^{a}(0)\|,   
\end{align*}
where the second inequality is to Assumption \ref{assump:W_follower}, i.e., 
\begin{align*}
\|\Wbf\Ybf^{a}\| = \left\|\Wbf\left(\Ibf-\frac{1}{n}\1\1^T\right)\Xbf^{a}\right\| \leq \sigma_{a}\|\Xbf^{a}\|.     
\end{align*}

\end{proof}
We note that this lemma only states that the followers in each cluster agree at a common point. However, these points might be different for different clusters, i.e., the nodes in different clusters might not agree with each other. Our next result is to study the consensus between the leaders under communication delays $\tau$. We note that under communication delays the estimate in \eqref{alg:leader} depends on the time interval $[k-\tau,k]$ for all $k\geq 0$. We, therefore, utilize the discrete-time variant of the $Gr\ddot{o}nwall$-$Bellman$ Inequality \cite{Khalil2002}, to hand such a dependence. The following lemma is to show that the leaders reach an agreement under some proper choice of step sizes. The analysis is motivated by the one in \cite{Doan2017}.


\begin{lem}\label{lem:leader_consensus}
Let $\beta\in\left(0,1-e^{-\frac{\ln(1/\delta_{C})}{\tau}}\right)$ and $\eta$ be defined as
\begin{align}
\eta \triangleq 1-\beta+\frac{\delta_{C}\beta}{(1-\beta)^{\tau}}\cdot\label{notation:eta}
\end{align}
Then the sequence $\{\Xbf_{\ell}(k)\}$ of the leaders satisfies
\begin{align}
\|\Xbf_{\ell}(k) - \1\xbar_{\ell}(k)^{T}\| \leq 2\eta^{k} \|\Xbf_{\ell}(0)\|.\label{lem:leader_consensus:ineq}    
\end{align}
\end{lem}

\begin{proof}
We denote by $\Ybf_{\ell} =  \Xbf_{\ell} - \1\xbar_{\ell}^T = (\Ibf-\frac{1}{n}\1\1^T)\Xbf_{\ell}$. Since $\Vbf(k)$ satisfies Assumption \ref{assump:V_leader}(c), using \eqref{alg:matrix_form} we have  
\begin{align*}
&\Ybf_{\ell}(k+1) = (1-\beta)\Xbf_{\ell}(k) + \beta\Vbf(k)\Xbf_{\ell}(k-\tau)\\
&\hspace{2cm} - (1-\beta)\1\xbar_{\ell}(k)^T + \beta \1\xbar_{\ell}(k-\tau)^T\\
&= (1-\beta)\Ybf_{\ell}(k) + \beta\Vbf(k)\Ybf_{\ell}(k-\tau)\\
&=(1-\beta)^{k+1}\Ybf_{\ell}(0) + \beta\sum_{t=0}^{k}V(t)\Ybf_{\ell}(t-\tau)(1-\beta)^{k-t},
\end{align*}
which by \eqref{notation:delta} implies that
\begin{align}
\|\Ybf_{\ell}(k+1)\| &\leq  (1-\beta)^{k+1}\|\Ybf_{\ell}(0)\|\notag\\ 
&\quad+ \beta\sum_{t=0}^{k}(1-\beta)^{k-t}\|V(t)\Ybf_{\ell}(t-\tau)\|\notag\\
&\leq (1-\beta)^{k+1}\|\Ybf_{\ell}(0)\|\notag\\ 
&\quad + \delta_{C}\beta\sum_{t=0}^{k}(1-\beta)^{k-t}\|\Ybf_{\ell}(t-\tau)\|, \label{lem:leader_consensus:eq1}
\end{align}
where the second inequality is due to Assumptions \ref{assump:V_leader} and \ref{assump:BConnectivity} 
\begin{align*}
\|V(t)\Ybf_{\ell}(t-\tau)\|\leq \sigma_{C}\|\Ybf_{\ell}(t-\tau)\|.
\end{align*}

We now apply the discrete-time variant of the $Gr\ddot{o}nwall$-$Bellman$ Inequality \cite{Khalil2002} to handle the right-hand side of \eqref{lem:leader_consensus:eq1}. Let $z(k)$ be defined as 
\begin{align*}
z(k) =  \sum_{t=0}^{k}(1-\beta)^{-t}\|\Ybf_{\ell}(t-\tau)\| 
\end{align*}
Thus we have $z(-1) = 0$ and $z(k)$ is a nondecreasing nonnegative function of time. Moreover, by \eqref{lem:leader_consensus:eq1} we have 
\begin{align*}
\|\Ybf_{\ell}(k+1)\| \leq  (1-\beta)^{k+1}\|\Ybf_{\ell}(0)\| + \delta_{C}\beta(1-\beta)^{k} z(k).  
\end{align*}
Consider 
\begin{align*}
&z(k+1) - z(k)= \left(1-\beta\right)^{-k-1}\|\Ybf_{\ell}(k+1-\tau)\|, 
\end{align*}
which implies that 
\begin{align*}
z(k+1) &= \left(1-\beta\right)^{-k-1}\|\Ybf_{\ell}(k+1-\tau)\| + z(k)\nonumber\\
&\leq (1-\beta)^{-\tau}\|\Ybf_{\ell}(0)\| + \frac{\delta_{C}\beta}{(1-\beta)^{\tau+1}}z(k-\tau) + z(k) \\
&\leq  (1-\beta)^{-\tau}\|\Ybf_{\ell}(0)\| + \left(1+ \frac{\delta_{C}\beta}{(1-\beta)^{\tau+1}}\right)z(k)\\
&\leq \frac{(1-\beta)\|\Ybf_{\ell}(0)\|}{\delta_{C}\beta} \left(1+\frac{\delta_{C}\beta}{\left(1-\beta\right)^{\tau+1}}\right)^{k+1},
\end{align*}
where we use $w(-1) = 0$ in the third inequality. Substituting the previous relation into \eqref{lem:leader_consensus:eq1} we have
\begin{align}
&\|\Ybf_{\ell}(k+1)\|\notag\\ 
&\leq (1-\beta)^{k+1}\|\Ybf_{\ell}(0)\|\nonumber\\  
&\quad + \|\Ybf_{\ell}(0)\|(1-\beta)^{k+1}\left(1+\frac{\delta_{C}\beta}{\left(1-\beta\right)^{\tau+1}}\right)^{k+1}\notag\\
&= (1-\beta)^{k+1}\|\Ybf_{\ell}(0)\|\notag\\ 
&\quad + \|\Ybf_{\ell}(0)\|\left(1 - \beta +\frac{\delta_{C}\beta}{\left(1-\beta\right)^{\tau}}\right)^{k+1}.\label{lem:leader_consensus:eq2}  
\end{align}
Since $\beta\in\left(0,1-e^{-\frac{\ln (1/\delta_{C})}{\tau}}\right)$ we have
\begin{align*}
\eta \triangleq 1-\beta+\frac{\delta_{C}\beta}{(1-\beta)^{\tau}}<1.
\end{align*}
Thus, by \eqref{lem:leader_consensus:eq2} and since $\|\Ybf_{\ell}\| \leq \|\Xbf_{\ell}\|$ we have \eqref{lem:leader_consensus:ineq}.
\end{proof}
Lemma \ref{lem:leader_consensus} states that under a proper choice of step size $\beta$, the leaders' iterates converge exponentially to the same value even under the presence of delays. On the other hand, Lemma \ref{lem:consensus_follower} states that the followers in each cluster agree at the same value. We now present a result to show that the followers' iterates in each cluster follows its leader' value.     


\begin{lem}\label{lem:leader_follower}
Let $P = \max_{i}\|p_{i}\|$. We have
\begin{align}
\|\xbar^{a}(k) - x_{\ell}^{a}(k)\| \leq  (1-\gamma)^{k}\|\xbar^{a}(0) - x_{\ell}^{a}(0)\| + \frac{2P\beta}{\gamma}.   \label{lem:leader_follower:ineq}
\end{align}
\end{lem}

\begin{proof}
We first note that $x_{i}(k) = p_{i}$ for all $i\in\Vcal$ and $k\in[-\tau,0]$. In addition, since each node $i\in\Vcal$ only considers the consensus updates \eqref{alg:follower} and  \eqref{alg:leader}, the nodes' iterates are always bounded, i.e., 
\begin{align}
\hspace{-0.5cm}\left\{\!\!\!
\begin{array}{l}
\|x_{i}^{a}(k)\| \leq (1-\gamma)P + \gamma P = P\\
\|x_{\ell}^{a}(k)\| \leq (1-\beta)P + \beta P = P
\end{array}\right.\!\!\!\!,\; \forall i\in\Vcal^{a}, \forall a, \forall k\geq 0.\label{lem:x_bound:ineq}   
\end{align}
Next, for any cluster $a$, by \eqref{lem:consensus_follower:eq1a} and \eqref{alg:leader} we have
\begin{align*}
\xbar^{a}(k+1) - x_{\ell}^{a}(k+1)&= (1-\gamma)(\xbar^{a}(k) - x_{\ell}^{a}(k))+ \beta x_{\ell}^{a}(k)\notag\\
&\quad - \beta \sum_{b\in\Ncal_{a}^{C}(k)}v_{ab}(k)x_{\ell}^{b}(k-\tau),
\end{align*}
which by using \eqref{lem:x_bound:ineq} yields
\begin{align*}
&\|\xbar^{a}(k+1) - x_{\ell}^{a}(k+1)\| \notag\\
&\leq (1-\gamma)\|\xbar^{a}(k) - x_{\ell}^{a}(k)\| + 2P\beta\notag\\
&\leq (1-\gamma)^{k+1}\|\xbar^{a}(0) - x_{\ell}^{a}(0)\| + 2P\beta\sum_{t=0}^{k}(1-\gamma)^{k-t}\notag\\
&\leq (1-\gamma)^{k+1}\|\xbar^{a}(0) - x_{\ell}^{a}(0)\| + \frac{2P\beta}{\gamma}\cdot
\end{align*}
\end{proof}
By putting the results in Lemmas \ref{lem:consensus_follower}--\ref{lem:leader_follower}, one can show that the nodes in the network achieve a consensus even under communication delays. The following theorem is to formally state this result.  

\begin{thm}\label{thm:main}
Let Assumptions \ref{assump:W_follower}--\ref{assump:BConnectivity} hold. Let the sequence $\{x_{i}^{a}(k)\}$ and $\{x_{\ell}^{a}(k)\}$, for all $i\in\Vcal^{a}$ and $a = 1,\ldots,r$, be generated by \eqref{alg:follower} and \eqref{alg:leader}. Then we have $\forall a$ and $\forall i\in\Vcal^{a}$
\begin{align}
\|x_{i}^{a}(k) - \xbar_{\ell}(k)\| &\leq \left((1-\gamma)\sigma_{a}\right)^{k}\|\Xbf^{a}(0)\| +  2\eta^{k}\|\Xbf_{\ell}(0)\|\notag\\ 
&\quad + (1-\gamma)^{k}\|\xbar^{a}(0) - x_{\ell}^{a}(0)\| + \frac{2P\beta}{\gamma}\cdot   \label{thm:ineq}
\end{align}
\end{thm}

\begin{proof}
We have
\begin{align*}
&x_{i}^{a}(k) - \xbar_{\ell}(k)\notag\\
&= x_{i}^{a}(k) - \xbar^{a}(k) + \xbar^{a}(k) - x_{\ell}^{a}(k) + x_{\ell}^{a}(k) - \xbar_{\ell}(k).
\end{align*}
Using the results in Lemmas \ref{lem:consensus_follower}--\ref{lem:leader_follower} and the triangle inequality immediately gives our result in \eqref{thm:ineq}.
\end{proof}

\begin{remark}
Here, we make some comments on the results given in Theorem \ref{thm:main}. 

1. As mentioned, to present the difference in the time scale between the dynamics of the followers and clusters, we choose $\beta \ll \gamma \in (0,1)$. Thus, Eqs.\ \eqref{lem:leader_consensus:ineq} and \eqref{thm:ineq} imply that the nodes converge arbitrarily close to each other exponentially. 

2. Since each graph $\Gcal^{a}$ is denser than the graph $\Gcal^{C}$ of the leaders, $\sigma_{a} \ll \sigma_{C} < 1$. By \eqref{thm:ineq} the exponential rate essentially depends on $\eta$, a function of $\sigma_{C}$ and the delays $\tau$. This show that after a transient time characterized by the topology of each $\Gcal^{a}$, the convergence  of the two-time-scale algorithms only depends on the connectivity of $\Gcal^C$ and the delays between leaders. This observation agrees with the one using singular perturbation theory \cite{Boker16,Awad19,JChow85,Biyik08}. We investigate further this observation numerically in Section \ref{sec:simulation}. 

3. One particular choice of the step sizes is $\gamma = \beta^{1/3}$. Under this choice, Eqs.\ \eqref{lem:leader_consensus:ineq} and \eqref{thm:ineq} shows that the nodes in the network reach a consensus at a rate $\beta^{2/3}$. This rate is similar to the one studied in (distributed) linear two-time-scale stochastic approximation; see for example \cite{Doan,Doan2020,Konda2004}. For example, let $T$ be a positive integer. If we run the algorithm in $T$ steps and let $\beta = 1/T$, then the algorithm converges at a rat $1/T^{2/3}$. 
\end{remark}

\section{Numerical Simulations}
\label{sec:simulation}
In this section, we investigate the impact of time delays and network structure on the performance of the two-time scale method, Algorithm \ref{alg:dist_two_time_scale}, over clustered networks. In particular, we use two different step sizes $\gamma,\beta$ in \textit{Algorithm 1} to represent the difference in information propagation between the nodes within each cluster and across different clusters, and two examples of clustered networks to see the effect of network structure to convergence speed. Since $\gamma$ is associated with the fast-time scale of the followers' updates in each cluster, we choose $\gamma \gg\beta$, where $\beta$ corresponds to the slow-time scale of the updates between the clusters. Moreover, $\beta$ is also chosen properly to handle the communication delays in the information exchange between the leaders. 

\subsection{Small network analysis}
\label{sim-1}
We first consider the performance of Algorithm \ref{alg:dist_two_time_scale} on a small network, i.e., a clustered network consisting of 60 agents divided into 3 clusters where nodes 1, 21, and 41 are chosen as leaders (see Fig. \ref{nw_image}). In each cluster $a$, a node is connected to the nearest two nodes, i.e., $|\Ncal^a_i =2|$ for all $i\in\Vcal^{a}$. Here, intra-cluster structure is not dense, which is purpose-built to compare with the dense cluster structure in Section \ref{sim-2}. 

\begin{figure}[htpb!]
    \centering
    \includegraphics[width=1\linewidth]{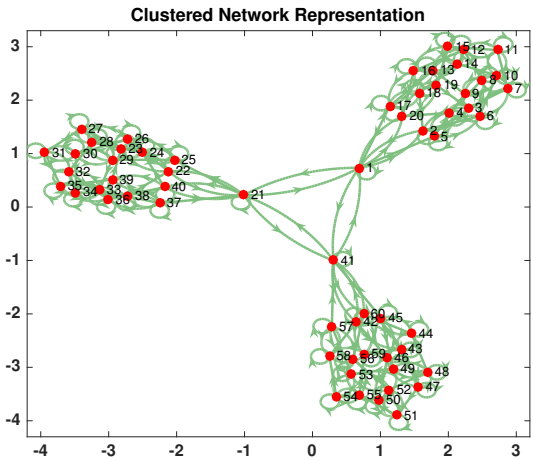}
    \caption{A 3-cluster network consisting of 60 nodes.}
    \label{nw_image}
\end{figure}

The adjacency matrix $\Wbf^a$ of each cluster $\Ccal^a$ is chosen as 
\begin{align*}\label{ex: Metropolis}
    \Wbf^a = [w^a_{ij}] = \left\{\begin{array}{lr}
        \frac{1}{|\Ncal^a_i|+1}, &\text{if} \; (i,j) \in \Ecal^a,\\
        0, &\text{if} \; (i,j) \notin  \Ecal^a, i \neq j,\\
        1 - \sum_{j \in \Ncal^a_i} w^a_{ij}, &\text{if} \; i = j.
    \end{array}\right.
\end{align*}
Similarly, we choose the adjacency matrix $\Vbf(k)$ corresponding to the sequence of graphs $\{\Gcal_{C}(k)\}$  between the clusters as above. It is straightforward to verify that the matrices $\Wbf^a$ and $\Vbf(k)$ satisfy Assumption 1 and 2. Finally, the initial conditions of all the nodes are chosen arbitrarily in $[-4, 4]$.


First, we investigate the performance of Algorithm \ref{alg:dist_two_time_scale} for two cases, namely $\beta = 1$ and $\beta = 0.1$ when the delays $\tau = 10$. That is, when $\beta = 1$ we simply consider a distributed consensus algorithm without handling the impact of delays $\tau$. As shown in Fig. \ref{nw_large_beta}, the consensus over the clustered network is not achieved, i.e., there is an oscillation among the nodes' iterates. This phenomenon is also observed in \cite{Magnus2020}. 

To handle this impact of time delays, we choose the step sizes $\beta= 0.1$ and $\gamma = 0.5$ which satisfy the condition in Lemma \ref{lem:leader_consensus}. In this case, all agents can achieve a consensus, as shown in Fig. \ref{nw_small_beta}. Moreover, it can be observed that the consensus within each cluster is achieved around $7s$, while the one between  clusters happens at $50s$. It means that the information shared inside each cluster is much faster than that between clusters, which agrees with our results in Theorem \ref{thm:main}.

\begin{figure}[htpb!]
    \centering
    \includegraphics[width=1\linewidth]{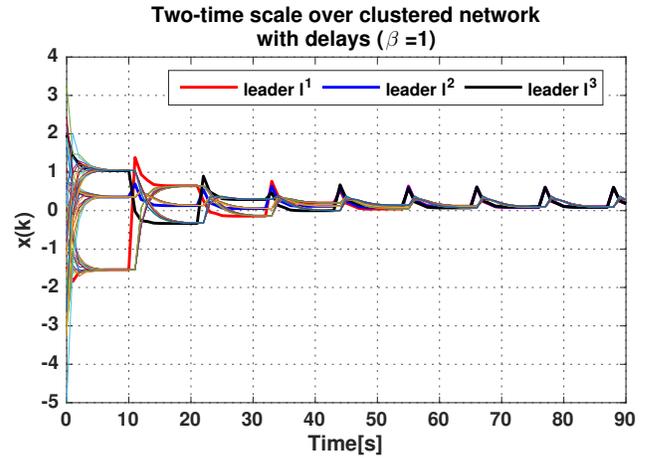}
    \caption{States of agents with $\tau = 10s$ and $\beta = 1$.}
    \label{nw_large_beta}
\end{figure}
\begin{figure}[htpb!]
    \centering
    \includegraphics[width=1\linewidth]{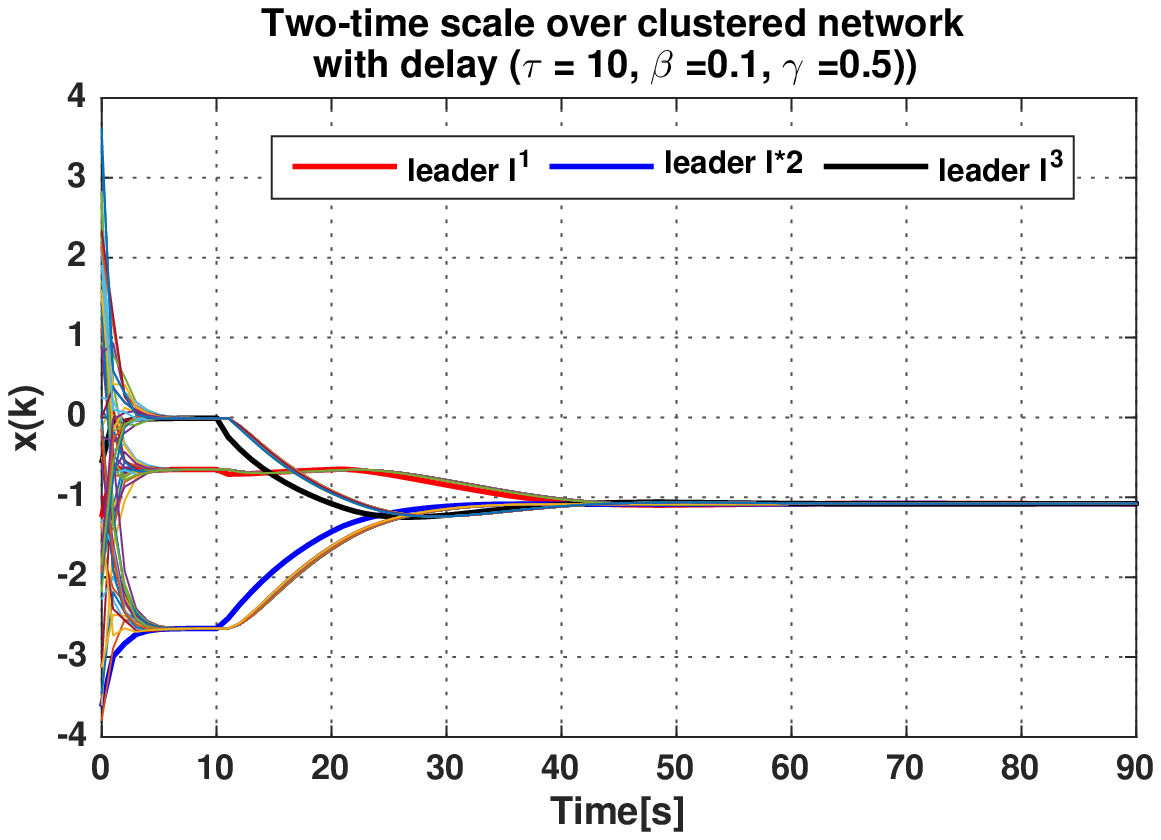}
    \caption{States of agents with $\tau = 10s$ and $\beta =0.1$.}
    \label{nw_small_beta}
\end{figure}
\begin{figure}[htpb!]
    \centering
    \includegraphics[width=1\linewidth]{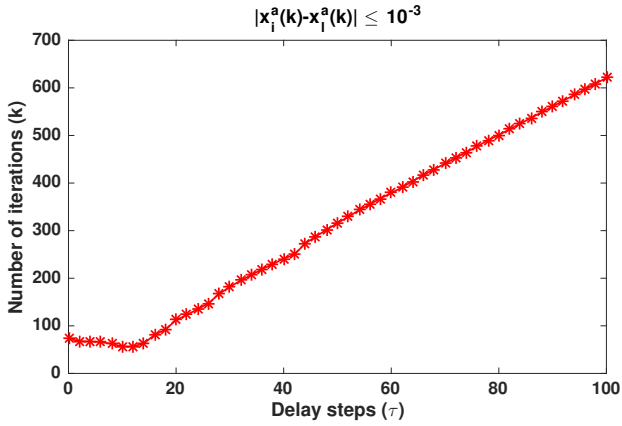}
    \caption{Number of iterations as a function of $\tau$}
    \label{nw_iterations}
\end{figure}

Next, impact of communication delays on convergence speed of the proposed algorithm is tested, where different time delays are used, and the stopping criterion is: $\|x_{i}^{a} - x_{\ell}^{a}\| \leq 10^{-3}$. As seen in  Fig. \ref{nw_iterations}, the number of iterations seems to depend linearly on $\tau$, which validates to our results shown in \textit{Theorem 1}.

\subsection{Larger network analysis}
\label{sim-2}
In this section, we validate the performance of Algorithm \ref{alg:dist_two_time_scale} for a larger network with much denser intra-cluster structure than that in Section \ref{sim-1}. More specifically, we consider a network of 400 nodes partitioned into 5 clusters, as depicted in Fig. \ref{nw_large network}. Each cluster structure $\Gcal^{a}$ is generated by randomly initializing positions of nodes and connecting two nodes if the distance between them is less than a certain value ($0.3$ in this simulation). The weighted adjacency matrices $\Wbf^{a}, \Vbf(k)$ are chosen similarly to that in Section \ref{sim-1}. The initial conditions of the nodes are chosen randomly  within $[-4,\; 4]$. Moreover, we set $\beta =0.05$ and $\gamma = 0.5$, which satisfies the condition in Lemma \ref{lem:leader_consensus}.

\begin{figure}[htpb!]
    \centering
    \includegraphics[width=1\linewidth]{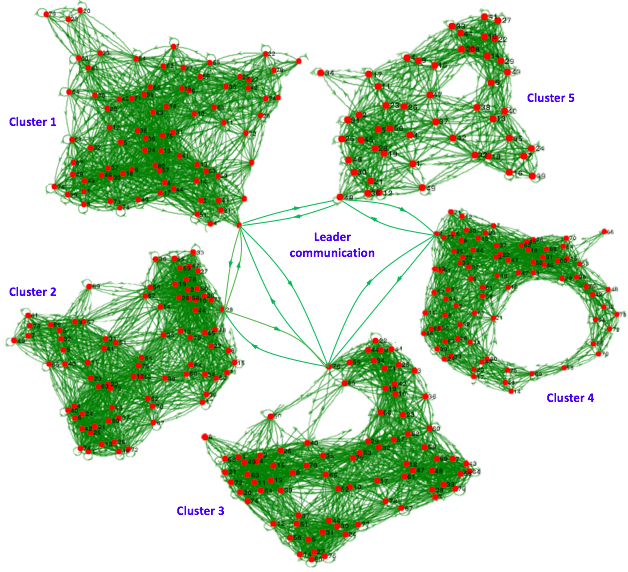}
    \caption{A 5-cluster network of 400 nodes.}
    \label{nw_large network}
\end{figure}

In this simulation, our goal is to investigate the impact of intra-cluster time delays to the proposed algorithm performance. To do so, we set the inter-cluster time delay to be $\tau = 20s$, and vary the intra-cluster time delay  $\tau_{a}$ to be 0, 2, and $15s$. 

Simulation results are then displayed in Fig. \ref{nw_consensuslarge}. As anticipated, when the intra-cluster delay is small compared to the inter-cluster delay, i.e., as $\tau_{a} = 0s$ or $\tau_{a} = 2s$, the convergence inside each cluster is reached much faster than that across clusters. Moreover, network convergence, both intra-cluster and inter-cluster, is faster as the intra-cluster delay is smaller.  On the other hand, when $\tau_{a}$ is large and comparable with $\tau$ (here $\tau_{a} = 15s$ and $\tau = 20s$), it greatly affects to the algorithm performance, where the convergence inside clusters cannot be distinguished from that between clusters. 

Note that most of the existing literature consider a uniform delay for all nodes, hence cannot capture the true behavior of clustered networks, where the delays within each cluster are much smaller as compared to the ones across the clusters. However, the simulations for our proposed distributed two-time-scale method here clearly show such behavior, and help provide a better understanding about the dynamical evolution in clustered networks.

\begin{figure}[htpb!]
    \centering
    \includegraphics[width=1\linewidth]{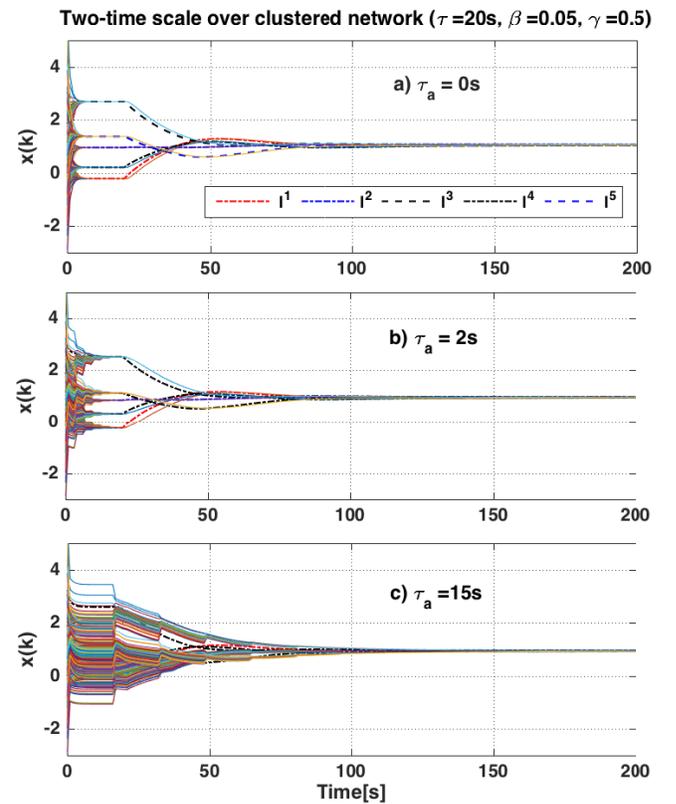}
    \caption{States of agents with $\beta =0.05, \gamma =0.5$, inter-cluster delay delay $\tau = 20s$, and varied intra-cluster delay $\tau_a$.}
    \label{nw_consensuslarge}
\end{figure}


\section{Conclusions and Future works}
In this paper, we consider a distributed two-time-scale consensus algorithm for clustered networks with inter-cluster time delays, where a faster consensus protocol is employed for each cluster while a slower consensus update is accounted for inter-cluster time delays. We proved the convergence of the two-time scale consensus algorithm in the presence of uniform, but possibly arbitrarily large, communication delays between the leaders.  In addition, we provided an explicit formula for the convergence rate of such algorithms, which characterizes the impact of delays and the network topology. Our theoretical results are validated by a number of numerical simulations. A few more interesting questions left from this work including handling quantized communication and directed graphs, which we leave for our future studies. 

\bibliographystyle{IEEEtran}
 \bibliography{refs}                                     

\begin{thebibliography}{10}
\providecommand{\url}[1]{#1}
\csname url@samestyle\endcsname
\providecommand{\newblock}{\relax}
\providecommand{\bibinfo}[2]{#2}
\providecommand{\BIBentrySTDinterwordspacing}{\spaceskip=0pt\relax}
\providecommand{\BIBentryALTinterwordstretchfactor}{4}
\providecommand{\BIBentryALTinterwordspacing}{\spaceskip=\fontdimen2\font plus
\BIBentryALTinterwordstretchfactor\fontdimen3\font minus
  \fontdimen4\font\relax}
\providecommand{\BIBforeignlanguage}[2]{{%
\expandafter\ifx\csname l@#1\endcsname\relax
\typeout{** WARNING: IEEEtran.bst: No hyphenation pattern has been}%
\typeout{** loaded for the language `#1'. Using the pattern for}%
\typeout{** the default language instead.}%
\else
\language=\csname l@#1\endcsname
\fi
#2}}
\providecommand{\BIBdecl}{\relax}
\BIBdecl

\bibitem{Romeres13}
D.~Romeres, F.~Dörfler, and F.~Bullo, ``Novel results on slow coherency in
  consensus and power networks,'' in \emph{Proc. of 2013 European Control
  Conference}, 2013, pp. 742--747.

\bibitem{Magnus2020}
R.~Funada, X.~Cai, G.~Notomista, M.~W.~S. Atman, J.~Yamauchi, M.~Fujita, and
  M.~Egerstedt, ``{Coordination of Robot Teams Over Long Distances: From
  Georgia Tech to Tokyo Tech and Back-An 11,000-km Multirobot Experiment},''
  \emph{IEEE Control Systems Magazine}, vol.~40, no.~4, pp. 53--79, 2020.

\bibitem{Bleibel18}
J.~Bleibel, M.~Habiger, M.~Lutje, F.~Hirschmann, F.~Roosen-Runge, T.~Seydel,
  F.~Zhang, F.~Schreiber, and M.~Oettel, ``Two time scales for self and
  collective diffusion near the critical point in a simple patchy model for
  proteins with floating bonds,'' \emph{Soft Matter}, vol. 14(8006), 2018.

\bibitem{Leiser17}
R.~J. Leiser and H.~G. Rotstein, ``Emergence of localized patterns in globally
  coupled networks of relaxationoscillators with heterogeneous connectivity,''
  \emph{Physical Review E}, vol. 96 (022303), 2017.

\bibitem{Proskurnikov17}
A.~V. Proskurnikov and R.~Tempo, ``A tutorial on modeling and analysis of
  dynamic social networks. part i,'' \emph{Annual Reviews in Control}, vol.~43,
  pp. 65--79, 2017.

\bibitem{ADas19}
A.~Das and A.~Lavina, ``Critical neuronal models with relaxed timescale
  separation,'' \emph{Physical Review X}, vol.~9, pp. 021\,062:1--021\,062:11,
  2019.

\bibitem{Leitch19}
J.~Leitch, K.~Alexander, and S.~Sengupta, ``Toward epidemic thresholds on
  temporal networks: a review and open questions,'' \emph{Applied Network
  Science}, vol. 4(105), 2019.

\bibitem{Barabasi-book06}
M.~Newman, A.-L. Barabasi, and D.~J. Watts, \emph{The Structure and Dynamics of
  Networks: (Princeton Studies in Complexity)}.\hskip 1em plus 0.5em minus
  0.4em\relax USA: Princeton University Press, 2006.

\bibitem{barabasi2016network}
\BIBentryALTinterwordspacing
A.-L. Barabási and M.~Pósfai, \emph{Network science}.\hskip 1em plus 0.5em
  minus 0.4em\relax Cambridge University Press, 2016. [Online]. Available:
  \url{http://barabasi.com/networksciencebook/}
\BIBentrySTDinterwordspacing

\bibitem{Mor2016}
I.-c. Morarescu, S.~Martin, A.~Girard, and A.~Muller-gueudin, ``Coordination in
  networks of linear impulsive agents,'' \emph{IEEE Transactions on Automatic
  Control}, vol.~61, no.~9, pp. 2402--2415, 2016.

\bibitem{Rahimian13}
M.~A. Rahimian and A.~G. Aghdam, ``Structural controllability of multi-agent
  networks: Robustness against simultaneous failures,'' \emph{Automatica},
  vol.~49, pp. 3149--3157, 2013.

\bibitem{Mousavi19}
H.~K. Mousavi, Q.~Sun, and N.~Motee, ``Measurable observations for network
  observability,'' in \emph{Proc. of 2019 American Control Conference (ACC)},
  2019, pp. 1438--1443.

\bibitem{Siami16}
M.~Siami and N.~Motee, ``Fundamental limits and tradeoffs on disturbance
  propagation in linear dynamical networks,'' \emph{IEEE Transactions on
  Automatic Control}, vol.~61, no.~12, pp. 4055--4062, 2016.

\bibitem{Pasqualetti13}
F.~Pasqualetti, F.~Dörfler, and F.~Bullo, ``Attack detection and
  identification in cyber-physical systems,'' \emph{IEEE Transactions on
  Automatic Control}, vol.~58, no.~11, pp. 2715--2729, 2013.

\bibitem{Boker16}
A.~M. Boker, C.~Yuan, F.~Wu, and A.~Chakrabortty, ``Aggregate control of
  clustered networks with inter-cluster time delays,'' in \emph{Proc. of 2016
  American Control Conference}, 2016, pp. 5340--5345.

\bibitem{Awad19}
A.~Awad, A.~Chapman, E.~Schoof, A.~Narang-Siddarth, and M.~Mesbahi,
  ``Time-scale separation in networks: State-dependent graphs and consensus
  tracking,'' \emph{IEEE Transactions on Control of Network Systems}, vol.~6,
  no.~1, pp. 104--114, 2019.

\bibitem{JChow85}
J.~Chow and P.~Kokotovic, ``Time scale modeling of sparse dynamic networks,''
  \emph{IEEE Transactions on Automatic Control}, vol.~30, no.~8, pp. 714--722,
  1985.

\bibitem{Biyik08}
E.~Biyik and M.~Arcak, ``Area aggregation and time-scale modeling for sparse
  nonlinear networks,'' \emph{Systems and Control Letters}, vol.~57, no.~2, pp.
  142--149, 2008.

\bibitem{XCheng18}
X.~Cheng and J.~M.~A. Scherpen, ``Clustering approach to model order reduction
  of power networks with distributed controllers,'' \emph{Advances in
  Computational Mathematics}, vol.~44, pp. 1917--1939, 2018.

\bibitem{SMartin16}
S.~Martin, I.-C. Morarescu, and D.~Nesic, ``Time scale modeling for consensus
  in sparse directed networks with time-varying topologies,'' in \emph{Proc. of
  2016 IEEE 55th Conference on Decision and Control (CDC)}, 2016, pp. 7--12.

\bibitem{Rejeb2015}
J.~B. Rejeb, I.~C. Morarescu, and J.~Daafouz, ``{Event triggering strategies
  for consensus in clustered networks},'' \emph{2015 European Control
  Conference, ECC 2015}, pp. 2156--2161, 2015.

\bibitem{Pham2020c}
\BIBentryALTinterwordspacing
V.~T. Pham, N.~Messai, D.~{Hoa Nguyen}, and N.~Manamanni, ``{Robust formation
  control under state constraints of multi-agent systems in clustered
  networks},'' \emph{Systems and Control Letters}, vol. 140, p. 104689, 2020.
  [Online]. Available: \url{https://doi.org/10.1016/j.sysconle.2020.104689}
\BIBentrySTDinterwordspacing

\bibitem{Pham2019f}
V.~T. Pham, N.~Messai, and N.~Manamanni, ``{Impulsive Observer-Based Control in
  Clustered Networks of Linear Multi-Agent Systems},'' \emph{IEEE Transactions
  on Network Science and Engineering}, vol.~7, no.~3, pp. 1840--1851, 2019.

\bibitem{Liu2017c}
\BIBentryALTinterwordspacing
X.~Liu, K.~Zhang, and W.~C. Xie, ``{Consensus seeking in multi-agent systems
  via hybrid protocols with impulse delays},'' \emph{Nonlinear Analysis: Hybrid
  Systems}, vol.~25, pp. 90--98, 2017. [Online]. Available:
  \url{http://dx.doi.org/10.1016/j.nahs.2017.03.002}
\BIBentrySTDinterwordspacing

\bibitem{Doan2017}
T.~T. Doan, C.~L. Beck, and R.~Srikant, ``On the convergence rate of
  distributed gradient methods for finite-sum optimization under communication
  delays,'' \emph{Proc. ACM Meas. Anal. Comput. Sys}, vol.~1, no.~2, pp.
  37:1--37:27, 2017.

\bibitem{DoanBS2018_ACC}
------, ``Convergence rate of distributed subgradient methods under
  communication delays,'' in \emph{{Proceedings of American Control Conference
  (ACC)}}, 2018.

\bibitem{Doan2018a}
T.~T. {Doan}, S.~T. {Maguluri}, and J.~{Romberg}, ``Distributed stochastic
  approximation for solving network optimization problems under random
  quantization,'' \emph{to appear on IEEE Transactions on Automatic Control},
  2020.

\bibitem{Doan}
\BIBentryALTinterwordspacing
T.~T. Doan, ``Finite-time analysis and restarting scheme for linear
  two-time-scale stochastic approximation,'' pp. 1--30, 2020. [Online].
  Available: \url{https://arxiv.org/abs/1912.10583}
\BIBentrySTDinterwordspacing

\bibitem{Doan2020}
T.~T. Doan and J.~Romberg, ``Finite-time performance of distributed
  two-time-scale stochastic approximation,'' \emph{Proceedings of Machine
  Learning Research}, vol. 120, pp. 1--11, 2020.

\bibitem{DoanM2019_Allerton}
T.~T. {Doan} and J.~{Romberg}, ``Linear two-time-scale stochastic approximation
  a finite-time analysis,'' in \emph{2019 57th Annual Allerton Conference on
  Communication, Control, and Computing (Allerton)}, 2019, pp. 399--406.

\bibitem{SMukherjee19}
S.~Mukherjee, H.~Bai, and A.~Chakrabortty, ``Block-decentralized model-free
  reinforcement learning control of two time-scale networks,'' in \emph{Proc.
  2019 American Control Conference (ACC)}, 2019, pp. 2233--2238.

\bibitem{Dalal2017a}
G.~Dalal, B.~Szorenyi, G.~Thoppe, and S.~Mannor, ``{Finite Sample Analysis of
  Two-Timescale Stochastic Approximation with Applications to Reinforcement
  Learning},'' \emph{Proceedings of Machine Learning Research}, vol.~75, pp.
  1--35, 2017.

\bibitem{pmlr-v99-karimi19a}
B.~Karimi, B.~Miasojedow, E.~Moulines, and H.-T. Wai, ``Finite time analysis of
  linear two-timescale stochastic approximation with {M}arkovian noise,'' in
  \emph{Conference on Learning Theory}, 2020.

\bibitem{GuptaSY2019_twoscale}
H.~Gupta, R.~Srikant, and L.~Ying, ``Finite-time performance bounds and
  adaptive learning rate selection for two time-scale reinforcement learning,''
  in \emph{Advances in Neural Information Processing Systems}, 2019.

\bibitem{MokkademP2006}
A.~Mokkadem and M.~Pelletier, ``Convergence rate and averaging of nonlinear
  two-time-scale stochastic approximation algorithms,'' \emph{The Annals of
  Applied Probability}, vol.~16, no.~3, pp. 1671--1702, 2006.

\bibitem{Konda2004}
V.~R. Konda and J.~N. Tsitsiklis, ``{Convergence rate of linear two-time-scale
  stochastic approximation},'' \emph{Annals of Applied Probability}, vol.~14,
  no.~2, pp. 796--819, 2004.

\bibitem{Olfati-Saber2004}
R.~Olfati-Saber and R.~M. Murray, ``{Consensus problems in networks of agents
  with switching topology and time-delays},'' \emph{IEEE Transactions on
  Automatic Control}, vol.~49, no.~9, pp. 1520--1533, 2004.

\bibitem{Ren2005}
W.~Ren and R.~Beard, ``{Consensus seeking in multiagent systems under
  dynamically changing interaction topologies},'' \emph{IEEE Transactions on
  Automatic Control}, vol.~50, no.~5, pp. 655--661, 2005.

\bibitem{Khalil2002}
H.~K. Khalil, \emph{Nonlinear System}, 3rd~ed.\hskip 1em plus 0.5em minus
  0.4em\relax Upper Saddle River, NJ: Prentice Hall, 2002.

\end{thebibliography}
\end{document}